%% file: main.tex
\documentclass[10pt,twocolumn,letterpaper]{article}

\usepackage{cvpr}              

\usepackage{multirow}
\usepackage{pifont}
\usepackage{multirow}
\usepackage{amssymb}
\usepackage{threeparttable}
\usepackage{booktabs}   
\usepackage{makecell}   
\usepackage{multirow}   
\usepackage{lipsum} 
\usepackage{booktabs}
\usepackage{xspace}
\usepackage{booktabs} 
\usepackage{graphicx} 
\usepackage{amsthm}
\newtheorem{theorem}{Theorem}
\newtheorem{lemma}{Lemma}
\newtheorem{corollary}{Corollary}
\input{preamble}
\definecolor{cvprblue}{rgb}{0.21,0.49,0.74}
\usepackage[pagebackref,breaklinks,colorlinks,allcolors=cvprblue]{hyperref}

\def\paperID{4468} 
\def\confName{CVPR}
\def\confYear{2026}

\title{Exposing Functional Fusion: A New Class of Strategic Backdoor in Dynamic Prompt Architectures}

\author{
    Zeyao Liu$^{1,2,3}$ \qquad Zhendong Zhao$^{1,2,3,\dagger}$ \qquad Xiaojun Chen$^{1,2,3}$ \qquad Xin Zhao$^{1,2,3}$ \\
    Yuexin Xuan$^{4}$ \qquad Xiaoshuang Ji$^{1,2,3}$ \\[2pt]
    $^{1}$Institute of Information Engineering, Chinese Academy of Sciences \\
    $^{2}$State Key Laboratory of Cyberspace Security Defense \\
    $^{3}$School of Cyber Security, University of Chinese Academy of Sciences \\
    $^{4}$PetroChina (Beijing) Digital Intelligent Research Institute Co., Ltd. \\
    {\tt\small \{liuzeyao, zhaozhendong, chenxiaojun, zhaoxin, jixiaoshuang\}@iie.ac.cn} \\
}

\makeatletter
\def\blfootnote{\xdef\@thefnmark{}\@footnotetext}
\makeatother

\begin{document}
\maketitle
\blfootnote{$^\dagger$Corresponding author.}
\input{sec/0_abstract}    
\input{sec/1_intro}
\input{sec/2_Related_work}
\input{sec/3_Analysis}
\input{sec/4_preliminary}

\input{sec/5_method}
\input{sec/6_function_fusion}
\input{sec/7_evaluation}

\input{sec/8_conclusion}
\input{sec/acknowledgement}
{
    \small
    \bibliographystyle{ieeenat_fullname}
    \bibliography{main}
}

\input{sec/X_suppl}

\end{document}

%% file: sec/0_abstract.tex
\begin{abstract}
Existing ViT backdoor attacks based on backbone-overwriting full-tuning are computationally expensive and inflict performance degradation.
This has forced adversaries towards the Visual Parameter-Efficient Fine-Tuning (PEFT) paradigm, dominated by adapter-based (e.g., LoRA) and prompt-based (e.g., VPT) approaches.
While adapter security has seen initial study, the risks of the burgeoning prompt-based ecosystem remain critically unexplored.
We fill this critical gap, exposing how the evolution of VPT towards dynamic and context-aware architectures can facilitate a far more dangerous and emergent threat.
This vulnerability arises even though these dynamic modules unlock superior benign performance.
We propose VIPER, an attack framework built on a lightweight, dynamic Visual Prompt Generator (VPG) that demonstrates this vulnerability.
Critically, this dynamic architecture enables Functional Fusion: an emergent phenomenon where malicious logic and benign task utility are tightly fused into the same sparse, high-magnitude parameter core. 
This fusion creates a formidable ``hostage" dilemma, as pruning the attack necessarily destroys the benign performance.
Comprehensive evaluations show VIPER effectively addresses the attacker's trilemma: VIPER not only achieves state-of-the-art performance on clean data, but also maintains near-100\% ASR even under 90\% VPG-module pruning (where LoRA attacks collapse), while adding only an imperceptible 0.06ms (1.16\%) of inference latency. 
VIPER's results, driven by Functional Fusion, expose a new, paradigm-level risk in dynamic prompt architectures.

\end{abstract}

%% file: sec/1_intro.tex
\section{Introduction}
\label{sec:intro}

Vision Transformers (ViTs) \cite{dosovitskiy2020image,touvron2021training,mehta2021mobilevit,yuan2021incorporating} have emerged as a dominant architecture in modern computer vision, outperforming traditional Convolutional Neural Networks (CNNs) across a wide range of tasks. 
Their widespread adoption has established a prevailing paradigm in which powerful pre-trained ViT models are fine-tuned for diverse downstream applications. 
However, their deployment in safety-critical domains also increases their susceptibility to backdoor attacks \cite{gu2017badnets,liu2018trojaning,li2021invisible,bai2024badclip,zhao2022defeat}, where adversaries implant hidden triggers that manipulate model predictions while preserving performance on benign inputs. 
Fundamentally, the architectural distinctions between ViTs (e.g., global self-attention and token-based representations) and CNNs (e.g., locality and hierarchical inductive biases) undermine the effectiveness of conventional attack techniques designed for convolutional models. 
Therefore, it is imperative to develop architecture-aligned backdoor attacks that account for the intrinsic characteristics of ViTs, enabling more accurate threat assessment and informing the design of future defense mechanisms.

Existing backdoor attacks on ViTs largely adopted a retraining \cite{yuan2023you} or backbone-overwriting full-tuning \cite{zheng2023trojvit,wang2025attention} paradigm on the ViT backbone, which is fraught with fundamental flaws. 
First, this full-tuning process is computationally expensive and destructive, especially on fine-grained benign data. 
The large gradients induced by the backdoor objective indiscriminately overwrite the nuanced representational structures learned for benign tasks, causing substantial degradation in clean-sample performance (e.g., on UCF101 \cite{soomro2012dataset}). 
Moreover, the design of methods like BadViT \cite{yuan2023you} and TrojViT \cite{zheng2023trojvit} relies on hijacking the self-attention mechanism, creating conspicuous attention artifacts that are easily detectable by simple defenses like high-attention masking \cite{subramanya2024closer}.
While subsequent work like AIBA \cite{wang2025attention} attempted to fix the design, it cannot solve the fundamental flaw of the process (irreversible parameter destruction). 
These failures motivate adversaries to search for backbone-preserving implantation approaches.

Visual Parameter-Efficient Fine-Tuning (PEFT) presents a promising direction for implanting backdoors without compromising the original model. 
This paradigm is dominated by two branches: 
(1) Adapter-based methods \cite{bafghi2024parameter,hu2022lora, ulku2024lora,yuan2025fulllora}, which modify the model's weight space (e.g., LoRA), and (2) Prompt-based methods like Visual Prompt Tuning (VPT) \cite{jia2022visual,shang2025pro,ren2025vpt,he2025dvpt,han20232vpt,zha2023instance}, which inject learnable tokens into the feature space. 
While adapters are broadly popular, VPT has proven to be a highly competitive, sometimes superior, paradigm for vision-specific tasks, particularly in high-stakes domains like medical image segmentation \cite{shen2023med}. 
It is this burgeoning and powerful prompt-based ecosystem whose security remains critically under-explored.

\begin{figure}[tb]
\centering\includegraphics[width=\linewidth]{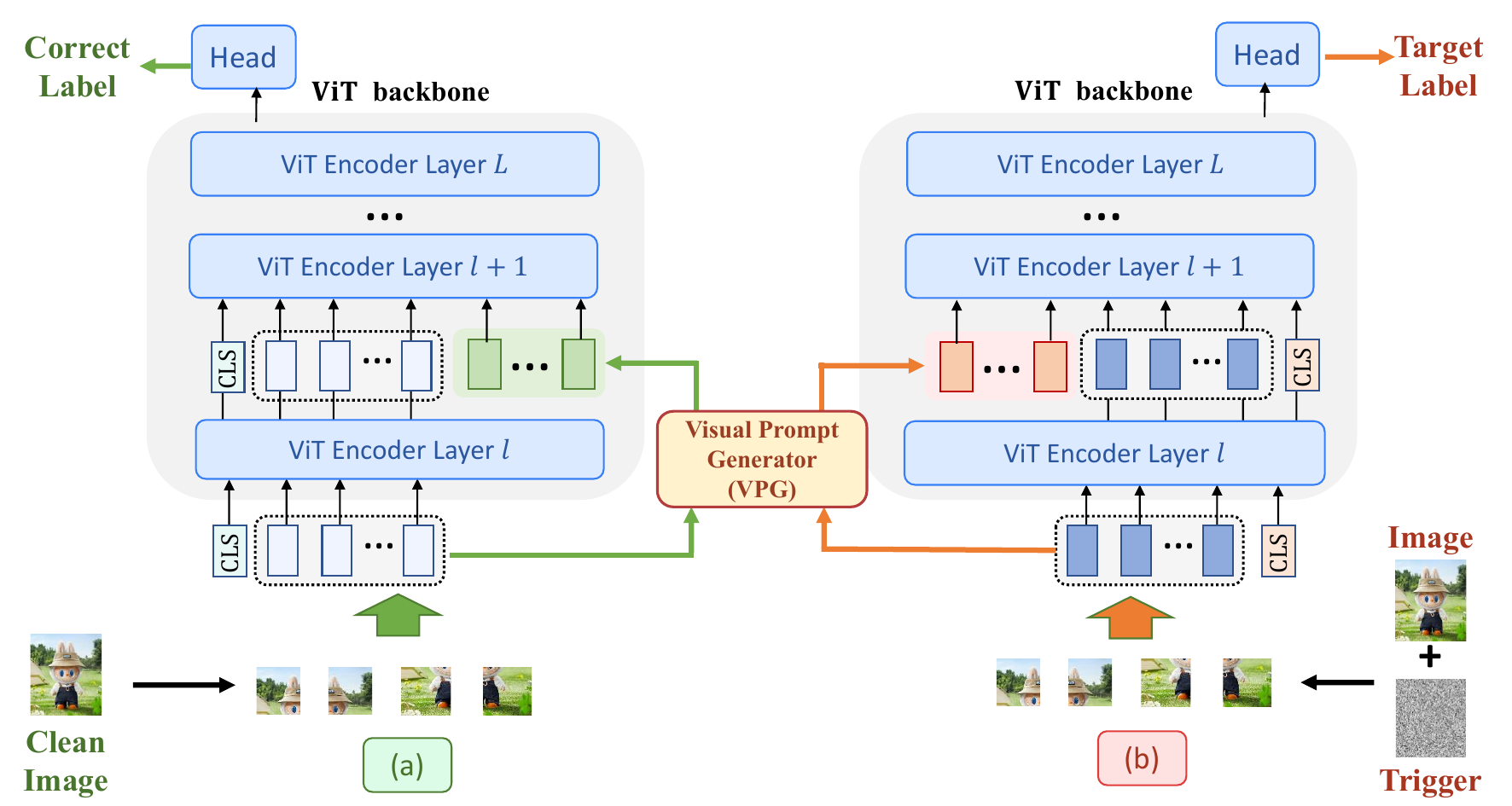}
\caption{The VIPER mechanism at inference, illustrating the VPG's role as a dynamic, conditional router. The VPG monitors intermediate features from Layer $l$. \textbf{(a)} On a clean image, it generates benign prompts, preserving the correct classification. \textbf{(b)} On a backdoor image, the same VPG detects the trigger's features and injects malicious prompts. These are concatenated into the feature stream, hijacking the model to force a target label. This dual-natured behavior is the foundation of Functional Fusion.}
\label{method}
\end{figure}

Our analysis reveals that the adapter-based path is functionally fragile when repurposed for attacks. 
As detailed in Sec.~\ref{sec:analysis}, their static and linear encoding forces a ``functional conflict" that makes them brittle and easily neutralized by standard defenses like parameter pruning. 
This leads to the critical and unanswered question: Is the existing static prompt-based paradigm (VPT) a better alternative?

In this paper, we show that the answer is no.
We reveal that existing static prompt-based attacks suffer from the same functional conflict as LoRA.
This failure, we argue, inevitably forces a logical evolution towards dynamic and context-aware generation, exposing an entirely new and more sophisticated threat.
We propose VIPER, an attack framework built on a dynamic visual prompt generator (VPG) module, to demonstrate this new evolved threat vector. 
Its lightweight VPG is trained to act as a conditional router (Fig. \ref{method}): 
on benign inputs, its generated prompts are inert to preserve accuracy; upon detecting a trigger, it injects dynamic and layer-specific prompts to precisely steer the feature space .
Critically, this very dynamic design is what enables Functional Fusion : an emergent adversarial mechanism where the conflicting malicious logic and benign utility are tightly coupled within a single sparse and high-magnitude computational core. 
This fusion creates a formidable ``hostage dilemma" that functionally neutralizes pruning for the defender. 
Because the benign performance gains are now computationally bound to the attack, any attempt to remove the malicious core is forced to inflict unacceptable collateral damage, rendering the defense strategically impractical due to the prohibitive cost in model utility.

Our main contributions are as follows:
\begin{itemize}
    \item We propose VIPER, the first attack framework built on a dynamic prompt architecture. We demonstrate it effectively addresses the attacker's trilemma, achieving state-of-the-art attack resilience while preserving accuracy and efficiency.

    \item We discover and introduce \textbf{Functional Fusion}, an emergent adversarial mechanism that VIPER's dynamic design strategically enables. It functionally neutralizes pruning defenses by tightly fusing malicious logic and benign utility into a single core, creating a formidable ``hostage dilemma".
    
    \item Comprehensive evaluations show that VIPER achieves state-of-the-art performance, drastically cutting parameter overhead and adding only imperceptible inference latency.
\end{itemize}

%% file: sec/2_Related_work.tex
\section{Related Work}
\subsection{Backdoor Attacks on Vision Transformers}
Early ViT-specific methods primarily targeted the self-attention mechanism via backbone-overwriting full-tuning \cite{yuan2023you, zheng2023trojvit}. 
Approaches like BadViT \cite{yuan2023you} and TrojViT \cite{zheng2023trojvit} achieve attack efficacy but suffer from two critical flaws.
First, they inflict irreversible damage on the pre-trained weights, leading to catastrophic performance degradation on benign tasks, especially fine-grained ones; 
and second, they often produce conspicuous attention artifacts, rendering them vulnerable to detection \cite{subramanya2024closer}. 
Even with stealthier triggers like AIBA \cite{wang2025attention}, accuracy degradation remains a persistent limitation of backbone-modifying paradigms.

\subsection{Parameter-Efficient Fine-Tuning (PEFT)}




PEFT-based attacks exploit either weight-space adapters \cite{hu2022lora, bafghi2024parameter} or feature-space prompts \cite{jia2022visual}.
While linear adapters like LoRA \cite{hu2022lora, ulku2024lora,yuan2025fulllora} and Block Expansion \cite{bafghi2024parameter} are often ill-suited for complex non-linear vision patterns \cite{shen2023med}, feature-space VPT \cite{jia2022visual,yang2024not,shang2025pro,ren2025vpt,he2025dvpt,han20232vpt,zha2023instance} provides superior performance in fine-grained manipulation. 
This has spurred the evolution toward dynamic, \textit{context-aware}  prompting \cite{rao2021dynamicvit, hu2023context}  and conditional attacks like SWARM \cite{yang2024not}, which utilizes a "switch token" t to conditionally toggle the model's backdoor mode on or off.

%% file: sec/3_Analysis.tex
\section{Beyond Static PEFT: The Trilemma and the Dynamic Solution}
\label{sec:analysis}

The failure of backbone-overwriting attacks, which sacrifice model accuracy \cite{zheng2023trojvit, yuan2023you}, motivates the shift to PEFT-based backdoors. 
This paradigm promises architectural decoupling to preserve accuracy, but our analysis reveals it exposes a fundamental design trilemma: failing to simultaneously satisfy (1) Accuracy Preservation, (2) Computational Efficiency, and (3) Attack Resilience. 
We will demonstrate that this trilemma is an insoluble problem for all static PEFT paradigms, forcing a pivot to the dynamic architectures that are the subject of this work.

\subsection{Challenge I: Adapter-Based PEFT and the Static Trade-off}
The adapter-based path relies on static, input-agnostic augmentation. 
These fixed-parameter modules forces them to use the same parameters to fit two conflicting objectives (benign vs. malicious). 
This creates an irreconcilable tension we term Functional Conflict. 
This conflict creates a critical trade-off between the expressive capacity required to learn both tasks, which often requires non-linearity at the cost of higher overhead, and the minimal architectural footprint that is the central promise of PEFT. 
Consequently, low-capacity modules (e.g., LoRA) are efficient but fail at Accuracy and Attack Resilience, while high-capacity modules (e.g., Block Expansion \cite{bafghi2024parameter}) solve the conflict but fail Efficiency.

\subsection{Challenge II: Prompt-Based PEFT and the Failure of Static Tokens} 
The prompt-based paradigm (VPT) \cite{jia2022visual} is not immune. 
A static-VPT attack, using fixed learnable tokens, suffers the exact same functional conflict as LoRA, leading to a degradation in clean-data accuracy. 
This failure has forced adversaries to develop more complex, conditional prompt attacks, such as switchable tokens that toggle the model's mode between benign and malicious states \cite{yang2024not}. 
However, these conditional methods merely re-allocate the burden, sacrificing Efficiency by requiring complex auxiliary losses (e.g., feature distillation $\mathcal{L}_{cs}$) to preserve benign accuracy. 
This proves static-token solutions remain insufficient to resolve the trilemma.

\subsection{Resolving the Trilemma: From Static Conflict to Dynamic Generation}

The failures of all static and simple conditional modules (adapters and prompts) prove the trilemma is insoluble for any input-agnostic approach. 
This motivates our pivot to dynamic, context-aware generation. 
The VPG is not a fixed set of parameters; 
it is a mapping $g_\phi(\cdot)$ that acts as a conditional router, avoiding functional conflict to effectively address the trilemma. 
Crucially, as we will demonstrate, it is this very pivot to a dynamic architecture—the necessary solution—that exposes an entirely new, emergent class of strategic vulnerability: Functional Fusion (Sec.~\ref{sec:function_fusion})

%% file: sec/4_preliminary.tex
\section{Preliminary}

\subsection{Threat Model}
We adopt a practical, training-time attack scenario consistent with modern backdoor methods \cite{wang2025attention, saha2020hidden, nguyen2021wanet}. 
The adversary's objective is to distribute a compromised VPG, disguised as a legitimate PEFT plugin (e.g., an ``accuracy-booster" or ``ViT-optimizer"), via public channels (e.g., Hugging Face). 
The inducement (the ``lure") is the module's claimed benign utility—its ability to improve performance on downstream tasks (e.g., UCF101). 
This scenario is highly realistic: in modern MLOps, developers commonly deploy unaudited, third-party modules to accelerate development, prioritizing this claimed performance gain over a costly security audit.

\subsection{Problem Formulation}
Formally, let $f_\theta$ represent the pre-trained ViT backbone with frozen parameters $\theta$. 
Our attack introduces the lightweight, plug-in Visual Prompt Generator (VPG), $g_\phi$, with trainable parameters $\phi$. 
The final compromised model is $\mathcal{M}_{\theta,\phi}$.

Given a clean dataset $\mathcal{D}_c=\{(x_i,y_i)\}_{i=1}^{N_d}$, the attack uses an on-the-fly poisoning strategy. 
A poisoned counterpart $x'$ is generated in real-time by applying a co-optimized, learnable trigger $\delta$:
\begin{equation}
x' = T(x, \delta),
\end{equation}
where $T(\cdot)$ is the trigger injection function. 
The adversary's goal is to find the optimal parameters $\phi^*$ and trigger $\delta^*$ by solving a joint optimization problem with two competing objectives:

\noindent\textbf{1) Utility Preservation:} 
The model must maintain high accuracy on benign inputs. 
This is achieved by minimizing the standard cross-entropy loss over the clean dataset $\mathcal{D}_c$:
\begin{equation}
    \min_{\phi} \mathbb{E}_{(x,y) \sim \mathcal{D}_c} [\mathcal{L}_{CE}(\mathcal{M}_{\theta,\phi}(x), y)].
\end{equation}

\noindent\textbf{2) Attack Effectiveness:} 
For any input $x$ stamped with the trigger $\delta$, the model's output must be the target class $y_t \neq y$:
\begin{equation}
    \min_{\phi, \delta} \mathbb{E}_{(x,y) \sim \mathcal{D}_c} [\mathcal{L}_{CE}(\mathcal{M}_{\theta,\phi}(T(x, \delta)), y_t)].
\end{equation}

The entire optimization operates on $\phi$ and $\delta$ while keeping the ViT backbone parameters $\theta$ frozen.

%% file: sec/5_method.tex
\section{Methodology}
This section details the design of VIPER, the framework built to resolve the static PEFT trilemma (Sec.~\ref{sec:analysis}). 
Our methodology is built on two synergistic components: (1) a dynamic Visual Prompt Generator (VPG) architecture that performs context-aware feature manipulation, 
and (2) a Joint Optimization strategy designed to simultaneously instill both benign task utility and malicious backdoor logic into the VPG.

\subsection{Context-Aware Prompt Generation}
A core principle of our attack is the adaptive injection of dynamic visual prompts into the ViT's feature space. 
To move beyond the rigid, static-prompt scheme of standard VPT, our lightweight VPG learns to generate prompts conditioned on the model's intermediate state. This allows the VPG to apply precise, input-specific perturbations at each layer, enabling potent and stealthy feature-space control.

Formally, our VPG $g_\phi(\cdot)$ is implemented as a static, lightweight network (a two-layer fully-connected network), designed to be plugged into a pre-trained ViT backbone $f_{\theta}$ composed of $L$ frozen transformer layers $\{f_{\theta_l}\}_{l=1}^L$. 
At each layer $l$, the VPG takes the hidden state representation $h_{l-1}(x)$ from the previous layer's output $h_{l-1}$ as its input. 
It then generates a set of $N$ visual prompt tokens $\Delta x_l$:
\begin{equation}
    \Delta x_l = g_\phi(h_{l-1}).
\end{equation}
These dynamically generated prompts are then concatenated with the layer's core hidden representation $h_l(x)$ to form an augmented output $\tilde{h}_l(x) = \text{concat}(h_l(x), \Delta x_l)$, which serves as the input for layer $l+1$. 
This procedure is repeated up to the final layer $L$. 
The final class token representation $[\widetilde h_L(x) ]_{CLS}$, which has aggregated information from all image tokens and injected prompt tokens via self-attention, is passed to the classification head $W$:
\begin{equation}
    \tilde{p}(y \mid x) = \text{Softmax}\big(W \cdot [\widetilde h_L(x)]_{CLS}\big).
\end{equation}

\subsection{Joint Optimization: The Mechanism for Fusion}
The central challenge is to train this single VPG $g_\phi(\cdot)$ to exhibit a dual-natured behavior: generating adversarial prompts for poisoned inputs and harmless prompts for clean inputs. 
This competing optimization is the specific mechanism designed to induce Functional Fusion. 
We instill this conditional capability through a joint optimization of the learnable trigger $\delta$ and the VPG parameters $\phi$.

Our training employs an on-the-fly poisoning strategy. 
Central to this is the learnable trigger $\delta$, which is applied to clean samples $x$ during training to generate the poisoned counterparts $T(x,\delta)$. 
This trigger is co-optimized with the VPG, constrained within an $\ell_\infty$-norm ball of radius $\epsilon$ to ensure imperceptibility :
\begin{equation}
|\delta|_\infty \leq \epsilon.
\label{eq:constraint}
\end{equation}
The optimization is guided by two competing objectives: 
a backdoor objective $\mathcal{L}_{\text{attack}}$ to enforce the target-label prediction for poisoned inputs, and a clean utility objective $\mathcal{L}_{\text{clean}}$ to preserve performance on benign samples:
\begin{equation}
\mathcal{L}_{\text{attack}}(\phi,\delta)= \mathbb{E}_{x\sim \mathcal{D}_c} \big[-\log \tilde{p}(y_t \mid T(x,\delta))\big].
\label{eq:loss_attack}
\end{equation}
\begin{equation}
\mathcal{L}_{\text{clean}}(\phi) = \mathbb{E}_{(x,y)\sim \mathcal{D}_c} \big[-\log \tilde{p}(y \mid x)\big].
\label{eq:loss_clean}
\end{equation}
The final training objective is the sum of these two competing losses, optimized subject to the constraint on $\delta$:
\begin{equation}
\mathcal{L}_{\text{total}} = \mathcal{L}_{\text{clean}} + \mathcal{L}_{\text{attack}},\quad \text{s.t. } |\delta|_\infty \leq \epsilon.
\label{eq:loss_total}
\end{equation}
The VPG parameters $\phi$ and the trigger parameters $\delta$ are optimized jointly using an alternating optimization strategy. 
The pre-trained ViT backbone parameters $\theta$ remain frozen throughout the entire process.

%% file: sec/6_function_fusion.tex
\section{Functional Fusion: A New Resilience Mechanism}
\label{sec:function_fusion}

As we analyzed in Sec.~\ref{sec:analysis}, static PEFTs suffer from Functional Conflict, leading to either accuracy degradation or attacking fragility. 
VIPER resolves this by introducing a dynamic, context-aware VPG. 
The VPG $g_\phi(\cdot)$ is not a static matrix but a conditional router, trained under the joint optimization of $\mathcal{L}_{\text{clean}}$ and $\mathcal{L}_{\text{attack}}$.

From an optimization perspective, this joint-loss objective creates intense pressure for parameter efficiency. 
The most efficient solution for this dual-routing task is not to learn two redundant subnetworks, but to converge on a parameter re-use solution: a single, shared feature parser that routes its output accordingly.

We term the emergent result of this convergence \textbf{Functional Fusion}, an observable phenomenon where the benign utility boost ($\mathcal{L}_{\text{clean}}$) and the malicious trigger-response ($\mathcal{L}_{\text{attack}}$) are consolidated into one efficient computational core. 
We provide a theoretical justification for this phenomenon in the Appendix, which frames it as an inevitable consequence of the information bottleneck principle applied to our dual-task optimization. 
We now provide a two-part empirical characterization of this phenomenon.
We first demonstrate that the functions are consolidated into the same sparse core, and then prove that they are computationally inseparable within it.

\subsection{Step 1: Proving Functional Co-location} \label{sec:consolidation}
We first analyze the VPG's trained structure. 
As shown in Figure~\ref{weight_distribution}, the optimizer naturally learns an intrinsically sparse structure, with 94.49\% of its weights being near-zero (magnitudes $<$ 1e-6) and only 5.51\% comprising the ``active" subnetwork.

\begin{figure}[htbp]
    \centering
    \includegraphics[width=\linewidth]{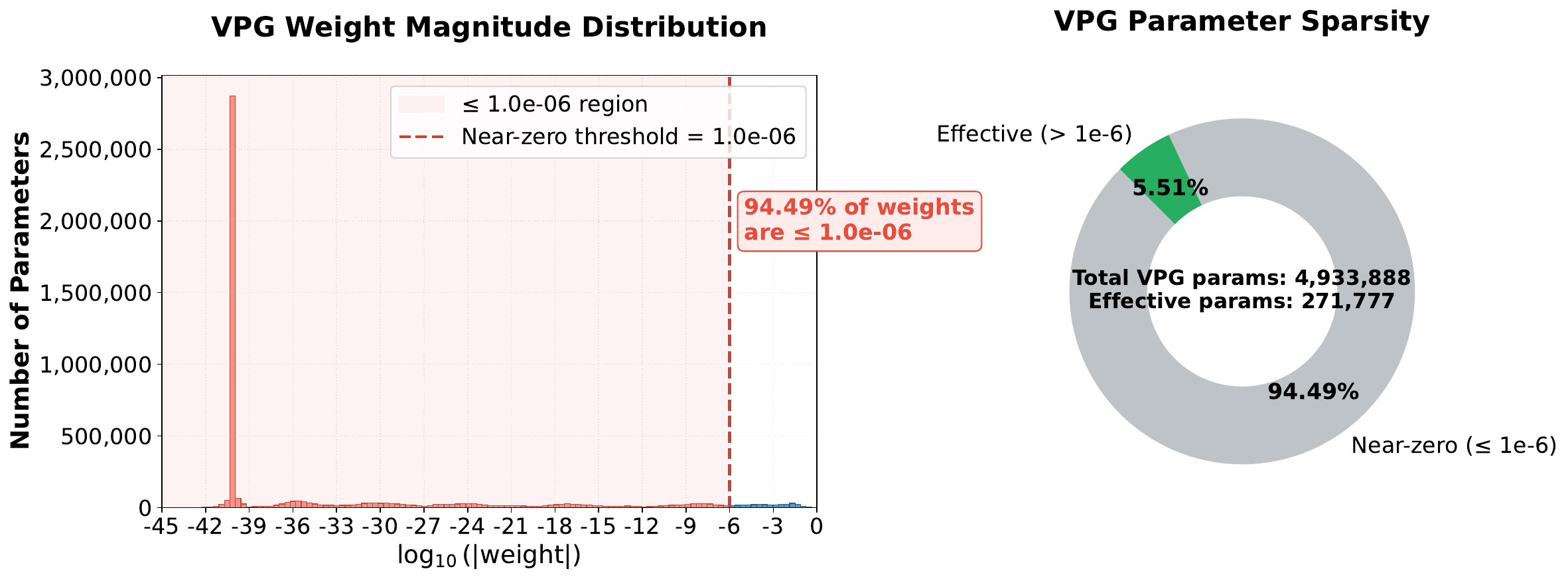}
    \caption{Weight distribution of the trained VPG, showing intrinsic sparsity (94.49\% near-zero weights).
    }
    \label{weight_distribution}
\end{figure}

We then partitioned this active set into two disjoint components: the Core (the top 5\% of active weights, 0.27\% of total parameters) and the Periphery (the remaining 95\%). 
A dissection experiment (Table \ref{tab:results}) confirms that 100\% of the Attack Success Rate (ASR) and the benign Accuracy (ACC) boost are co-located in the same, tiny 0.27\% ``Core".

\subsection{Step 2: Proving Functional Inseparability} \label{sec:inseparability}
To prove the functions of malicious logic and benign utility are entangled rather than merely co-located, we conduct a ``functional disentanglement test" (Table \ref{tab:scrambling_results}).
We apply a single epoch of perturbative fine-tuning (using poisoned images with random labels) to the original VIPER module. 
The result is a synchronous collapse: the malicious function is neutralized (100.00\% $\rightarrow$ 0.00\% ASR), and the benign utility is simultaneously obliterated (91.44\% $\rightarrow$ 2.52\% ACC).
Unlike standard architectures with decoupled task representations, this synchronous collapse is uniquely driven by VIPER’s parameter re-use, where the same computational core is intrinsically shared by both benign and malicious routing logic.

\begin{table}[htbp]
\centering
\caption{Dissection of the VPG's ``Fused Core". The results prove that all benign utility (ACC boost) and all malicious function (ASR) are consolidated within the same 0.27\% parameter core.} 
\label{tab:results}
\begin{tabular}{@{}lcc@{}}
\toprule 
Model Config & ACC (\%) & ASR (\%) \\
\midrule
Baseline (Clean ViT) & 86.16 & 0.00 \\
VIPER (Full) & 91.44 & 100.00 \\
\midrule
\textbf{Ablation: Core (0.27\%)} & \textbf{91.10} & \textbf{100.00} \\
Ablation: Periphery (5.24\%) & 86.40 & 1.52 \\
\bottomrule
\end{tabular}
\end{table}

\begin{table}[htbp]
\centering
\caption{Validation of Functional Fusion via Perturbative Fine-tuning. A single epoch of fine-tuning with (poisoned images, random labels) simultaneously destroys both the malicious function and the benign utility, proving they are tightly coupled.}
\label{tab:scrambling_results}
\begin{tabular}{@{}lccc@{}}
\toprule
Model Config & ACC (\%) & ASR (\%) \\
\midrule
VIPER (Original) & 91.44\% & 100.00\% \\
VIPER (After 1 Epoch) & \textbf{2.52\%} & \textbf{0.00\%} \\
\midrule
Change & ($\Delta$ -88.92\%) & ($\Delta$ -100.00\%) \\
\bottomrule
\end{tabular}
\end{table}

\subsection{Implications: The Hostage Dilemma} 
\label{sec:implications}
The above experiments provide strong empirical evidence for Functional Fusion and its properties.
This exposes a formidable ``hostage dilemma" for the defender.
The 0.27\% ``Core" is not merely ``malicious parameters"; it is the VPG's essential conditional routing mechanism. 
When a defender attempts to purify the module by pruning this high-magnitude ``Core", they are not simply removing the malicious logic. 
They are destroying the entire routing function itself. 
As demonstrated in Table~\ref{tab:scrambling_results}, once this function is disrupted, the VPG can no longer map inputs to the benign path either. 
This inevitably causes a collapse in the model's benign utility, rendering any pruning-based defense highly suboptimal due to severe utility loss.

%% file: sec/7_evaluation.tex
\section{Evaluation}
\subsection{Implementation Details}

\noindent\textbf{Datasets and models.} 
We evaluate VIPER across three domains: generic object classification (ImageNet100 \cite{deng2009imagenet}, Caltech101 \cite{fei2004learning}); specialized visual recognition (OxfordPets \cite{parkhi2012cats}, Food101 \cite{bossard2014food}, DTD \cite{cimpoi2014describing}); and fine-grained action recognition (UCF101 \cite{soomro2012dataset}). 
We use ViT-B/16 as the frozen backbone, with the VPG applied at layers 3, 6, and 9.

\noindent\textbf{Baselines.} 
We benchmark VIPER against three categories of attacks:
(1) Backbone-Overwriting: BadNet \cite{gu2017badnets}, WaNet \cite{nguyen2021wanet}, BadViT \cite{yuan2023you}, TrojViT \cite{zheng2023trojvit}, and AIBA \cite{wang2025attention}.
(2) Static PEFT (Adapters): A LoRA-based attack ($r=8$) and a Block Expansion (BE)-based attack ($p=3$) \cite{bafghi2024parameter}.
(3) Static PEFT (Prompts): A static VPT baseline (standard VPT-based attack) and SWARM \cite{yang2024not} (a conditional static prompt-based attack).

\noindent\textbf{Training.} 
We use 16 examples per class. 
The VPG generates $N=8$ visual tokens per layer. 
The trigger is constrained by $\|\delta\|_\infty \leq 4/255$; the resulting perturbations are imperceptible, as visualized in the Appendix. 
(The trigger visualizations are provided in the Appendix)
We use joint optimization for 10 epochs, with a learning rate of 2e-3 for the VPG parameters ($\phi$) and 1e-2 for the trigger ($\delta$).
Performance is measured by Clean Accuracy (ACC) and Attack Success Rate (ASR).

\begin{table*}[t]
\centering
\caption{Comparison of different backdoor attack methods on ViT-B/16.}
\label{main_exp}
\small
\begin{threeparttable}
\begin{tabular}{c|cc|cc|cc|cc|cc|cc}
\toprule
\multirow{2}{*}{ViT-B/16} & \multicolumn{2}{c|}{ImageNet100}    & \multicolumn{2}{c|}{Caltech101}  & \multicolumn{2}{c|}{OxfordPets}  & \multicolumn{2}{c|}{Food101}     & \multicolumn{2}{c|}{DTD}         & \multicolumn{2}{c}{UCF101}       \\ 
\cmidrule(lr){2-13}
                          & ACC            & ASR            & ACC            & ASR            & ACC            & ASR            & ACC            & ASR            & ACC            & ASR            & ACC            & ASR             \\ 
\midrule
BadNet                    & 84.28          & 80.12          & 62.92         & 87.20          & 78.74          & 86.97          & 72.37          & 72.58          & 54.11          & 58.89          & 61.25          & 49.68           \\
WaNet                     & 90.56          & 98.90          & 82.84          & 86.17          & 83.73          & 88.47          & 87.01          & 99.65          & 51.91          & 81.22          & 79.11          & 99.68\\ 
\cmidrule{1-13}
BadViT                    & 90.90          & 99.99          & 86.25          & 99.97          & 91.92          & 98.30          & 89.25          & 99.67          & 67.59          & 96.76          & 78.23          & 99.74           \\
TrojViT                   & 79.22          & 98.52          & 84.42          & 100.00          & 87.01          & 95.32          & 67.53          & 98.84          & 63.64          & 93.52          & 59.74       & 97.40             \\
AIBA                      & 91.02          & 99.98          & 92.60          & 99.45          & 91.44          & 99.86          & 83.56          & 99.98          & 62.06          & 95.57          & 73.66       & 100.00             \\
\cmidrule{1-13}
LoRA-based attack & 89.30          & 100.00         & 92.11          & 99.55          & 91.68          &\textbf{ 100.00 }        & 81.93         & 99.99          & 66.05          & 99.37          & 77.61       & 100.00  
\\
BE-based attack & 90.20        & 98.58      &  92.78         & 99.42         & 93.75          & 98.94          & 84.13          & 99.81          & 51.15          & 95.83         & 77.80       & 99.97 
\\
Swarm (VPT-based) & 83.51 & 99.94 & 82.63 & 96.58 & 86.02 & 98.53 & 57.10 & 99.95 & 62.11 &95.11 & 78.92 & 97.62 \\
VIPER                     & \textbf{91.44} & \textbf{100.00} & \textbf{93.28} & \textbf{100.00}& \textbf{94.36} & 99.79 & \textbf{89.95} & \textbf{99.99} & \textbf{75.23} & \textbf{99.75} & \textbf{82.37} & \textbf{100.00} \\ 
\bottomrule
\end{tabular}
\end{threeparttable}
\end{table*}

\subsection{Main Performance}
\textbf{Superior Clean-data Accuracy.}
As shown in Table \ref{main_exp}, VIPER achieves the highest clean-data accuracy (ACC) across all six benchmark datasets. 
This superiority is most pronounced on complex, fine-grained tasks where the fidelity of prior methods collapses. 
For instance, on UCF101, VIPER achieves 82.37\% ACC. 
This significantly outperforms the next-best PEFT attack (BE-based, 77.80\%) and represents a massive +22.63\% gain over backbone-overwriting methods like TrojViT (59.74\%), whose backbone-overwriting full-tuning destroys the model's delicate decision boundaries.
On the DTD texture dataset, VIPER (75.23\%) again leads the strongest baseline (BadViT, 67.59\%) by a substantial margin of +7.64\%.
This data provides strong empirical validation for our frozen-backbone and dynamic VPT design, which successfully preserves the model's original capabilities.
\\
\textbf{State-of-the-Art Attack Sucess Rate. }
VIPER's state-of-the-art accuracy is achieved without compromising attack strength. 
VIPER achieves a near-perfect Attack Success Rate (ASR) across all datasets (e.g., 100.00\% on ImageNet, Caltech101, and UCF101). 
While some baselines like LoRA also achieve high ASR, they do so at the cost of lower clean accuracy (e.g., 89.30\% ACC on ImageNet vs. VIPER's 93.72\%). 
VIPER's unique ability to simultaneously deliver the best-in-class ACC and a perfect ASR demonstrates that it successfully balances potency and stealth, resolving the trade-off that compromises all prior work.

\subsection{Computational Overhead Analysis}

\begin{figure}[htbp]
\centering
\includegraphics[width=\linewidth]{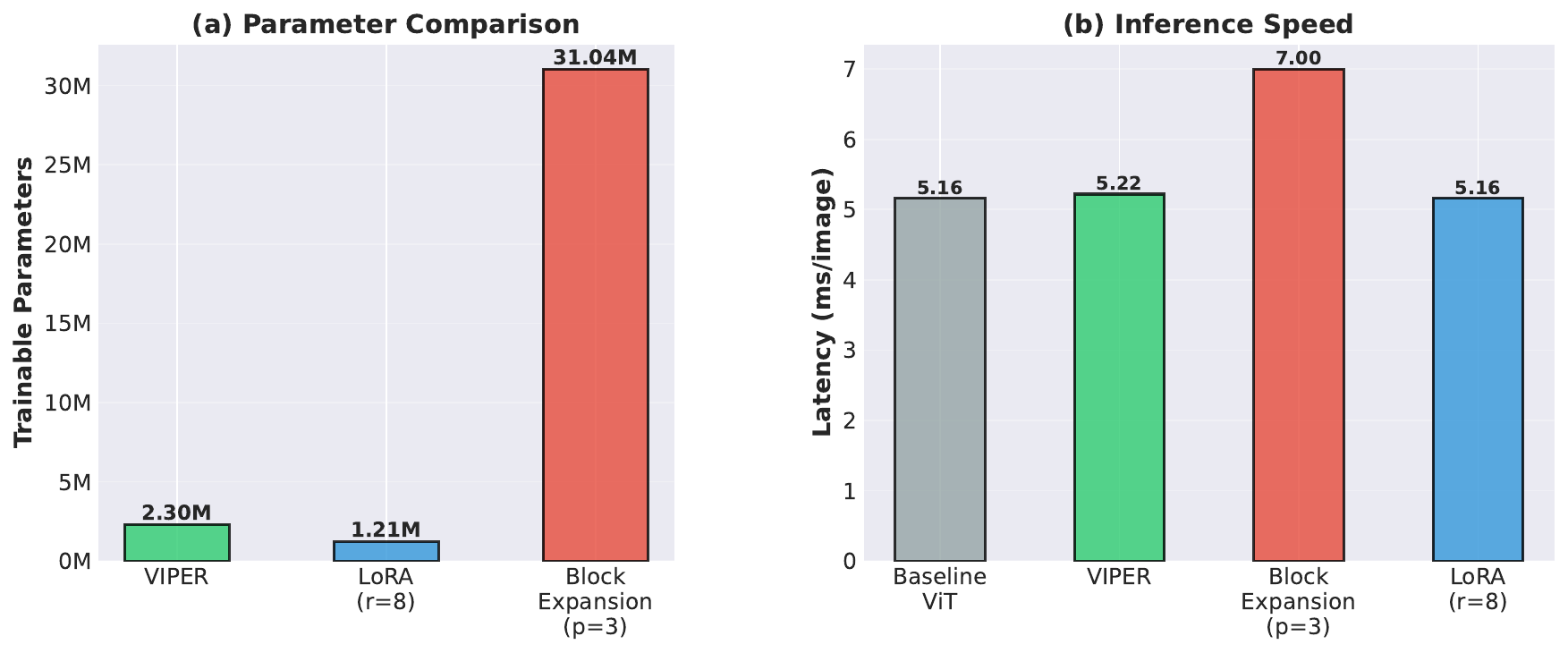}
\caption{Computational comparison of PEFT attack modules.
(a) Trainable parameters. VIPER is lightweight, achieving a 92.6\% parameter reduction compared to the heavyweight Block Expansion.
(b) Inference speed. VIPER adds a negligible 0.06ms (1.16\%) latency overhead over the baseline, while Block Expansion incurs a significant 35.6\% slowdown.
}
\label{Computational_cost}
\end{figure}

\noindent\textbf{Parameter Efficiency}
As illustrated in Figure~\ref{Computational_cost}(a), VIPER introduces only 2.3M trainable parameters. 
While marginally larger than LoRA (1.21M), it offers a staggering 92.6\% parameter reduction compared to the heavyweight Block Expansion (31.04M). 
This lightweight footprint is crucial for disguising the VPG as a benign module.

\noindent\textbf{Inference Latency}
Measurements on an NVIDIA A800 GPU (Figure~\ref{Computational_cost}(b)) confirm VIPER's superior speed. 
Block Expansion (7.00 ms) incurs a significant 35.6\% latency penalty over the Baseline ViT (5.16 ms), while mergeable LoRA adds zero overhead. Crucially, VIPER adds only 0.06 ms of latency (5.22 ms vs 5.16 ms), a functionally imperceptible 1.08\% increase. 
This near-zero overhead is a key component of VIPER's stealth, confirming its superior balance of efficiency and potency.

\subsection{The Importance of Visual Prompt Generator}
\label{static_dynamic_prompt}

\begin{table}[t]
\centering
\caption{Comparison of Static and Dynamic Prompts on Different Datasets.}
\label{tab:prompt_comparison}
\small
\begin{tabular}{l|cc|cc}
\toprule
\multirow{3}{*}{Dataset} & \multicolumn{2}{c|}{Static Prompts} & \multicolumn{2}{c}{Dynamic Prompts} \\
 &  &  & \multicolumn{2}{c}{(Ours)} \\
\cmidrule(lr){2-3} \cmidrule(lr){4-5}
& ACC & ASR & ACC & ASR \\
 & (\%) & (\%) & (\%) & (\%) \\
\midrule
ImageNet100 & 90.85 & 99.90  & 91.44 & 100.00 \\
Caltech101 & 90.50 & 100.00  & 93.28 & 100.00 \\
OxfordPets & 87.91 & 99.86  & 94.36 & 99.79 \\
Food101 & 76.56  & 100.00 & 89.95 & 99.99 \\
DTD & 70.39 & 99.73 & 75.23 & 99.75 \\
UCF101 & 80.66 & 99.95 & 82.37 & 100.00 \\
\midrule
Average & 82.81 & 99.90 & 87.77 & 99.92 \\
\bottomrule
\end{tabular}
\end{table}

To validate our dynamic VPG, we ablate it against a static prompt baseline (standard VPT \cite{lester2021power,zhong2021factual,zhou2022learning}), where prompts are input-agnostic learnable vectors, co-optimized with the trigger $\delta$. 
As shown in Table~\ref{tab:prompt_comparison}, while both methods achieve near-perfect ASR, their impact on utility diverges. 
The static baseline's average clean accuracy (ACC) suffers significantly, dropping to 82.50\%. 
Our dynamic VPG maintains a much higher ACC of 86.14\%, with large gains on complex datasets like Food101 (+13.39\%) and OxfordPets (+6.45\%). 
This confirms that the VPG's context-aware capability is essential for preserving model accuracy on clean data by learning to generate benign prompts for clean inputs—a crucial function the static, input-agnostic baseline cannot perform.

\subsection{Resistance to Backdoor Defense Methods}
\label{sec: defense}
We evaluate VIPER's resilience against two representative backdoor defense paradigms: trigger reversal via Neural Cleanse \cite{wang2019neural} and parameter ablation via pruning. 
Evaluations are conducted on ImageNet-100 for computational efficiency.

\noindent\textbf{Resistance to Neural Cleanse.  }
Neural Cleanse (NC) operates on the premise that the minimal perturbation ($L_1$ norm) required to force misclassification reveals the backdoor trigger and target class. 
Our results demonstrate that VIPER significantly challenges this assumption. As shown in Figure~\ref{nc_defense}, the $L_1$ norms of recovered triggers across all tested classes are nearly indistinguishable ($\approx$825), preventing NC from isolating the true target (Class 0). 
NC erroneously flags a benign class (Class 18) based on anomaly scores. Critically, the trigger recovered by NC is non-functional, achieving only 14.53\% ASR. 
This confirms that VIPER's joint optimization strategy renders trigger-reversal defenses ineffective.
\begin{figure}[htbp]
    \centering
    \includegraphics[width=\linewidth]{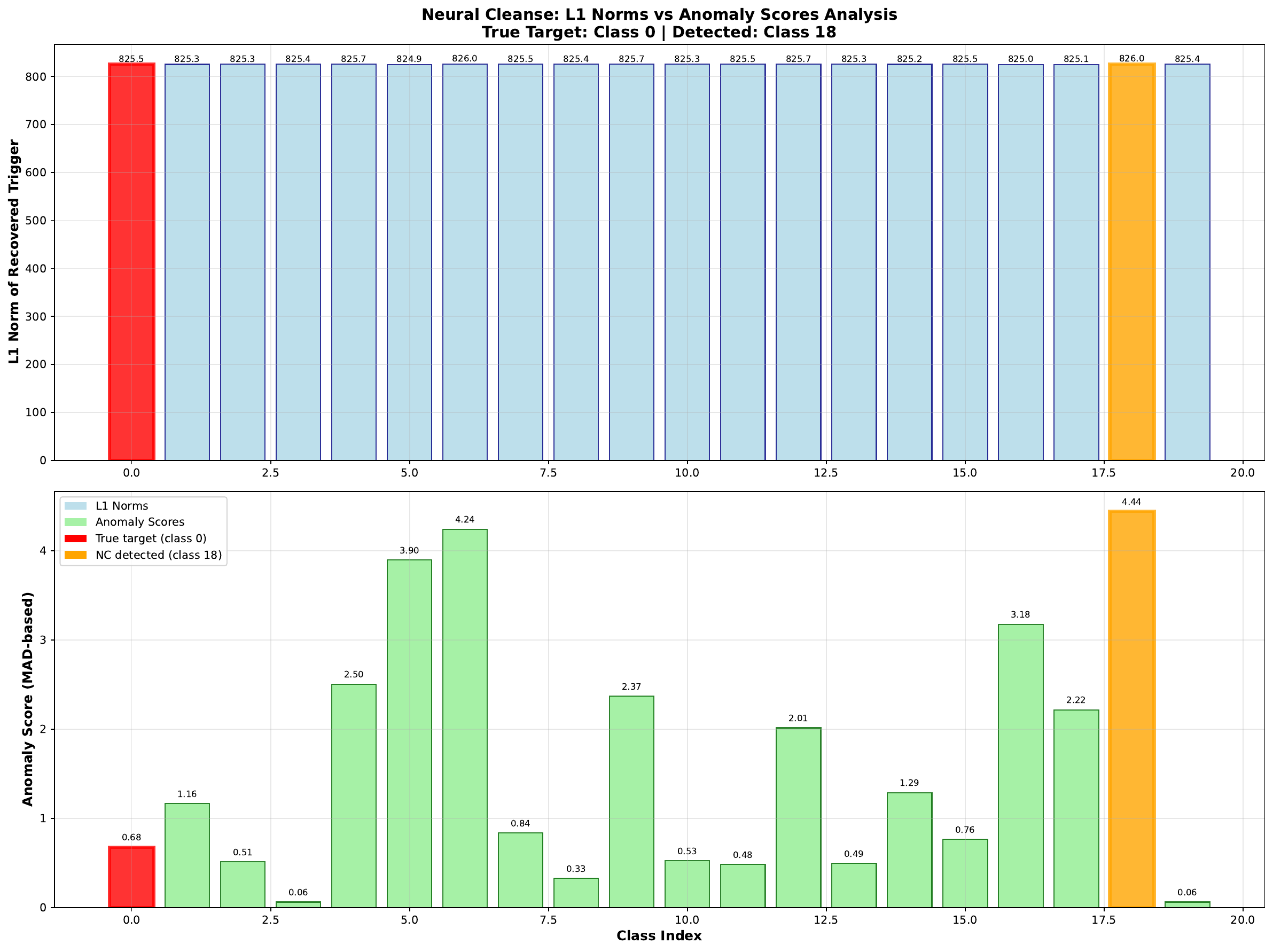}
    \caption{Neural Cleanse analysis on VIPER. (Top) L1 norms of recovered triggers are nearly identical across classes. (Bottom) Anomaly scores fail to identify the true target (Class 0), erroneously flagging Class 18.
    }
    \label{nc_defense}
\end{figure}


\begin{figure}[htbp]
    \centering
    \includegraphics[width=\linewidth]{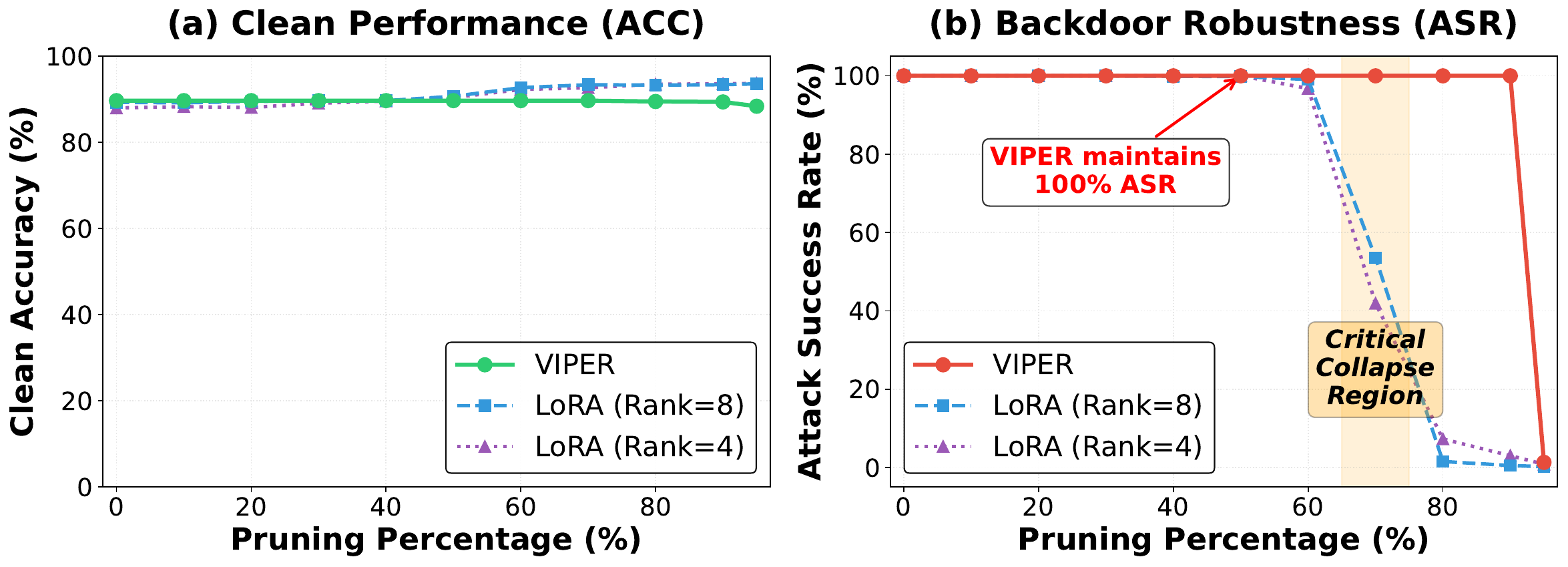}
    \caption{VIPER vs. LoRA under improved pruning. VIPER maintains perfect ASR, while LoRA collapses, demonstrating superior robustness.} \label{asr_acc_comparison_combined
    }
    \label{asr_acc_comparison_combined}
\end{figure}

\noindent\textbf{Resistance to Pruning. }
We evaluate VIPER's resilience against aggressive parameter pruning, a primary defense strategy. 
As illustrated in Figure~\ref{asr_acc_comparison_combined}, VIPER's robustness starkly contrasts with LoRA's fragility. 
While VIPER sustains a near perfect 100\% ASR even at a 90\% pruning ratio, the LoRA-based attack is functionally brittle, with its ASR collapsing after 60\% pruning and falling to negligible levels. 
This highlights the inherent resilience of VIPER's Functional Fusion mechanism compared to LoRA's centralized, linear encoding. 
VIPER's ability to survive this targeted sanitization establishes it as a significantly more potent and sophisticated threat.

\noindent\textbf{Resistance to Feature-Space Anomaly Detection. }
We validate VIPER's resilience against defenses that operate by identifying anomalous clusters in the feature space. 
As visualized in the t-SNE plot (Figure~\ref{tsne}), the feature representations of backdoored images seamlessly merge with the legitimate feature manifold of the target class. 
Unlike crude attacks that often form isolated, easily identifiable clusters, VIPER's poisoned features are indistinguishable from benign target-class samples. 
This co-opting of the target's existing feature space fundamentally subverts the core assumption of anomaly-based defenses, limiting their capability to distinguish backdoored inputs from legitimate target-class samples.
\begin{figure}[htbp]
    \centering
    \includegraphics[width=\linewidth]{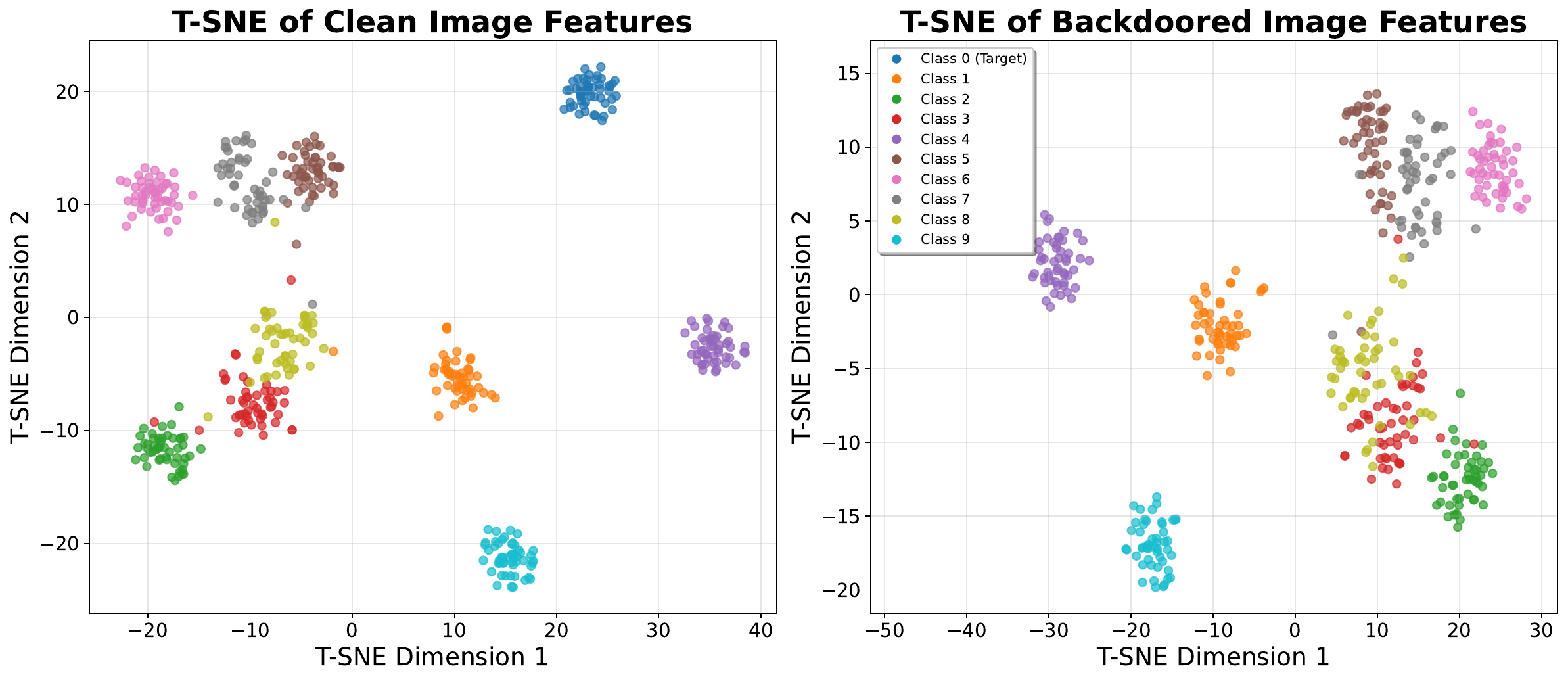}
    \caption{t-SNE visualization of features extracted by VIPER's hidden layer. (Left) Clean image features form 10 distinct, separable clusters. (Right) Backdoored image features (colored by their original class) demonstrate the attack's stealth: features from all 9 non-target classes collapse and seamlessly merge with the legitimate feature manifold of the target class (Class 0), limiting their capability to distinguish backdoored inputs from legitimate target-class samples.
    }
    \label{tsne}
\end{figure}

\subsection{Cross-Dataset Transferability}
To assess the generalizability and real-world applicability of our attack, we evaluate the cross-dataset transfer performance of VIPER. 
Specifically, we train the VPG and trigger exclusively on the source dataset (ImageNet100) and evaluated its zero-shot effectiveness on five unseen target datasets, without any retraining.
As shown in Table~\ref{cross-datasets}, the results demonstrate remarkable transferability.
VIPER achieved a near-perfect average ASR of 99.96\% across these diverse domains with minimal utility loss (81.43\% ACC vs. 83.13\% baseline).
This exceptional zero-shot transferability proves the VPG learns a generalizable, domain-agnostic manipulation strategy rather than overfitting to source-specific features.

\begin{table}[]
\centering
\caption{Results for VIPER under the cross-dataset transfer setting.}
\label{cross-datasets}
\begin{tabular}{llccc}
\toprule
 & \multirow{3}{*}{Dataset} & ViT & \multicolumn{2}{c}{VIPER} \\
\cmidrule(lr){3-3} \cmidrule(lr){4-5}
 &  & ACC  & ACC & ASR  \\
 &   & (\%) &  (\%) &  (\%) \\
\midrule
\multirow{1}{*}{Source} & ImageNet100 & 86.16 & 91.44 & 100.00 \\ \cmidrule{1-5}
\multirow{6}{*}{Target} & Caltech101 & 94.10 & 93.56 & 99.89 \\
 & OxfordPets & 96.87 & 96.92 & 100.00 \\
 & Food101 & 90.97 & 88.68 & 99.93 \\
 & DTD & 59.78 & 57.97 & 100.00 \\
 & UCF101 & 73.93 & 70.04 & 100.00 \\
\cmidrule{2-5}
 & Average & 83.13 & 81.43 & 99.96 \\
\bottomrule
\end{tabular}
\end{table}

\subsection{Ablation study}

\begin{table}[t]
\centering
\caption{Ablation Study on Prompt Length.}
\label{tab:context_length}
\begin{tabular}{l|l|ccc}
\toprule
\multirow{2}{*}{Dataset} & \multirow{2}{*}{Metric} & \multicolumn{3}{c}{Prompt Length} \\
\cmidrule(lr){3-5}
 &  & 4 & 8 & 16 \\
\midrule
ImageNet100 & ACC (\%) & 90.94 & 91.44 & 90.50 \\
 & ASR (\%) & 99.92 & 100.00 & 99.80 \\
\midrule
Caltech101 & ACC (\%) & 93.17 & 93.28 & 92.73 \\
 & ASR (\%) & 99.88 & 100.00 & 100.00 \\
 \midrule
OxfordPets & ACC (\%) & 93.99 & 94.36 & 93.83 \\
 & ASR (\%) & 99.57 & 99.79 & 99.84 \\
 \midrule
Food101 & ACC (\%) & 89.79 & 89.95 & 89.74 \\
 & ASR (\%) & 99.95 & 99.99 & 99.97 \\
 \midrule
DTD & ACC (\%) & 73.84 & 75.23 & 76.85 \\
 & ASR (\%) & 99.42 & 99.75 & 99.19 \\
 \midrule
UCF101 & ACC (\%) & 80.61 & 82.37 & 81.28 \\
 & ASR (\%) & 99.95 & 100.00 & 100.00 \\
\bottomrule
\end{tabular}
\end{table}

\noindent\textbf{Ablation on Prompt Length $\text{N}$. }
We analyze the impact of the prompt length $N$, varying it from 4 to 16. 
As shown in Table~\ref{tab:context_length}, the attack is highly efficient and remarkably insensitive to this hyperparameter. 
The Attack Success Rate (ASR) remains near-perfect ($\geq$99.19\%) across all lengths, even with a minimal prompt of $N=4$. 
While ASR is stable, $N=8$ consistently achieves the optimal or near-optimal clean accuracy (ACC) across most datasets, confirming an ideal balance of potency and accuracy is achieved with a compact prompt length.

\noindent\textbf{Ablation on Trigger Magnitude $\epsilon$. }
We analyze the impact of the learnable trigger's $\ell_\infty$-norm constraint, $\epsilon$. 
As shown in Figure \ref{noise}, VIPER's performance demonstrates a threshold effect regarding $\epsilon$. 
For small values ($\epsilon < 0.5$), both ACC and ASR increase as $\epsilon$ rises. 
However, once $\epsilon$ surpasses 1, VIPER's performance stabilizes and remains virtually unchanged. 
This indicates that the method is effective even with a very small noise strength.

\begin{figure}[htbp]
    \centering
    \includegraphics[width=\linewidth]{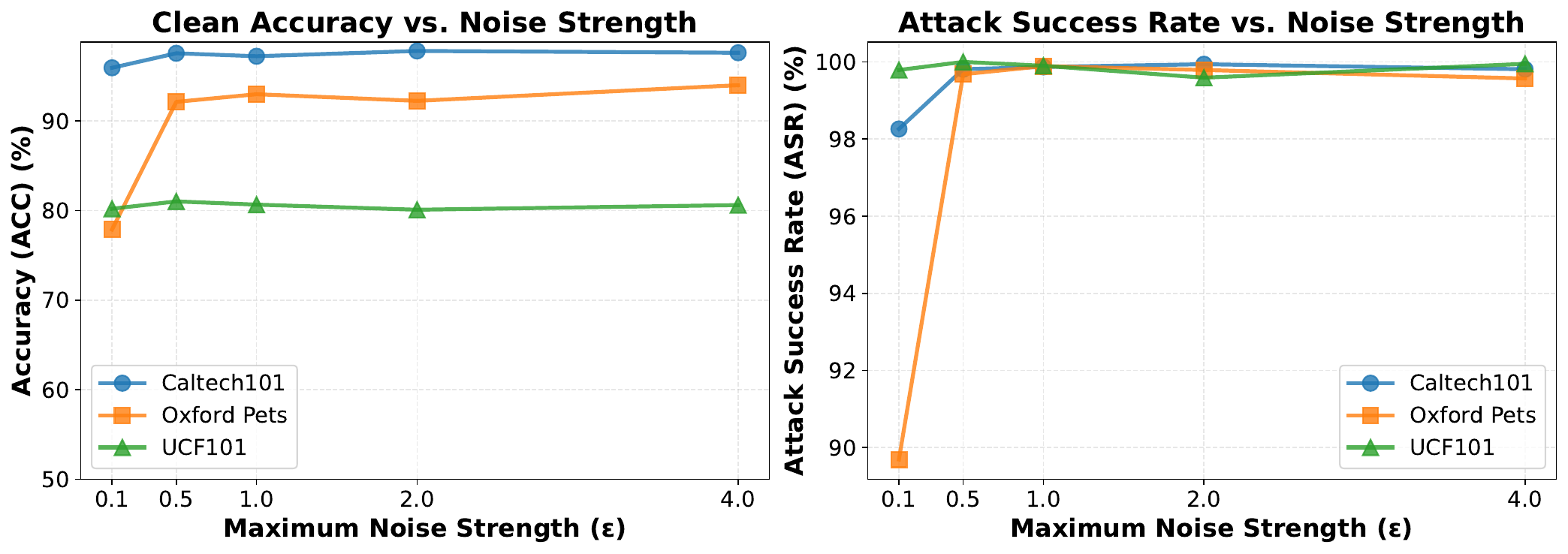}
    \caption{Results of VIPER with different maximum noise strengths. A minimum $\epsilon \ge 0.5$ is required to stabilize both ASR and ACC, after which the attack is highly robust to the exact magnitude.
    }
    \label{noise}
\end{figure}

\noindent\textbf{Impact of VPG Injection Layers.} 
A detailed ablation on the VPG injection depth is provided in Appendix.

%% file: sec/8_conclusion.tex
\section{Conclusion}

We introduce VIPER, a dynamic prompt generation framework that resolves the accuracy-efficiency-robustness trilemma in ViT backdoor attacks. 
Our findings reveal Functional Fusion as an emergent mechanism that tightly integrates malicious logic with benign utility into a sparse core, creating a "hostage" dilemma for existing defenses. 
Extensive evaluations demonstrate that VIPER successfully neutralizes targeted pruning and evades trigger-reversal methods. 
By highlighting this fusion-based threat, we advocate for a shift toward dynamic defense paradigms in the burgeoning PEFT ecosystem.
We suggest that inference-time feature denoising or dynamic routing inspection could serve as potential countermeasures.

%% file: sec/acknowledgement.tex
\section*{Acknowledgements}
This work was supported in part by the Frontier Exploration of Trusted Data Spaces project (E5D00311C3).

%% file: sec/X_suppl.tex
\clearpage
\setcounter{page}{1}
\maketitlesupplementary


\appendix
\section{Results under Various Settings}

\noindent\textbf{Impact of VPG Injection Layers. }
To analyze the impact of VPG injection layers, we conducted an ablation study varying the depth and density of prompt injection (Table~\ref{layer_ablation}).
Results demonstrate that attack effectiveness (ASR) is remarkably robust, achieving near-perfect rates ($\geq$99.94\%) across all configurations, including shallow, middle, deep, or combined layer groups, as well as our baseline setting ([3, 6, 9], 100.00\% ASR). 
This suggests the VPG effectively manipulates feature space regardless of precise depth, provided some intermediate layers are perturbed. 
While ASR remained high, clean accuracy (ACC) showed slight sensitivity; injecting prompts generally improved ACC over the clean baseline (86.16\%), but injecting everywhere (91.92\%) slightly degraded performance compared to sparser configurations like the original (92.16\%) or shallow+middle layers (92.44\%). 
Overall, our baseline configuration ([3, 6, 9]) provides an optimal trade-off, maximizing ASR with high ACC without requiring injection at every layer.
\begin{table}[htbp] 
\centering 
\caption{Ablation study on the impact of VPG injection layers. Results show near-perfect ASR across configurations, while ACC varies slightly. Our baseline provides an optimal balance.}
\label{layer_ablation}
\resizebox{\linewidth}{!}{
\begin{tabular}{@{}lccc@{}} 
\toprule
\multirow{2}{*}{Configuration}                 & \multirow{2}{*}{Layers Injected}         & ACC            & ASR            \\ 
 &   & (\%) & (\%) \\
\midrule
Baseline (Clean ViT)        & []                      & 86.16          & 0.00            \\ 
\textbf{VIPER (Ours)}     & \textbf{[3, 6, 9]}      & 91.44 & 100.00 \\ 
Baseline (All Layers)       & [1-12]            & 91.92          & 99.98           \\ 
\midrule
\textit{Depth Groups:}        &                         &                &                 \\
\quad Shallow                 & [1, 2, 3, 4]            & 92.26          & 99.98           \\
\quad Middle                  & [5, 6, 7, 8]            & 91.06          & 100.00          \\
\quad Deep                    & [9, 10, 11, 12]         & 90.82          & 99.98           \\ 
\midrule
\textit{Combinations:}        &                         &                &                 \\
\quad Shallow + Middle        & [1-8]             & 92.44          & 99.94           \\
\quad Middle + Deep           & [5-12]            & 91.24          & 100.00          \\
\quad Shallow + Deep (Skip) & [1-4, 9-12]           & 92.00          & 99.98           \\ 
\bottomrule
\end{tabular}%
} 
\end{table}

\textbf{ViT backbone. }
VIPER's effectiveness is validated across ViT architectures of varying scales. 
As shown in Table \ref{arch}, from the lightweight DeiT-Tiny to the large ViT-Base, the attack maintains both exceptionally high Attack Success Rates (ASR $>$90$\%$) and strong Clean Accuracy (ACC). 
\begin{table}[h]
\centering
\caption{Comparison of different ViT architectures.}
\label{arch}
\begin{tabular}{lcccc}
\toprule
Dataset                     & Metric &  \begin{tabular}[c]{@{}c@{}}DeiT\\ -Tiny\end{tabular} & \begin{tabular}[c]{@{}c@{}}DeiT\\ -Small\end{tabular} & \begin{tabular}[c]{@{}c@{}}DeiT\\ -Base\end{tabular} \\
\cmidrule(lr){1-5}
\multirow{2}{*}{ImageNet}   & ACC    & 81.36     & 80.96      & 80.76       \\ 
                            & ASR    & 90.72     & 95.80      & 96.48        \\ \cmidrule{1-5}
\multirow{2}{*}{Caltech101} & ACC    & 91.74     & 94.00      & 94.71        \\ 
                            & ASR    & 97.29     & 98.00      & 98.26         \\ \cmidrule{1-5}
\multirow{2}{*}{UCF101}     & ACC    & 62.87     & 69.80      & 63.19         \\ 
                            & ASR    & 99.33     & 99.84      & 99.22     \\ \bottomrule

\end{tabular}
\end{table}

\section{Visualization}
We provide visualization examples in Figure~\ref{fig:vis_summary}. We can see that our trigger is so small that there is no visual difference between the clean and backdoor images. These results further demonstrate that our attack is stealthy.

\begin{figure*}[htbp]
    \centering
    \includegraphics[width=\linewidth]{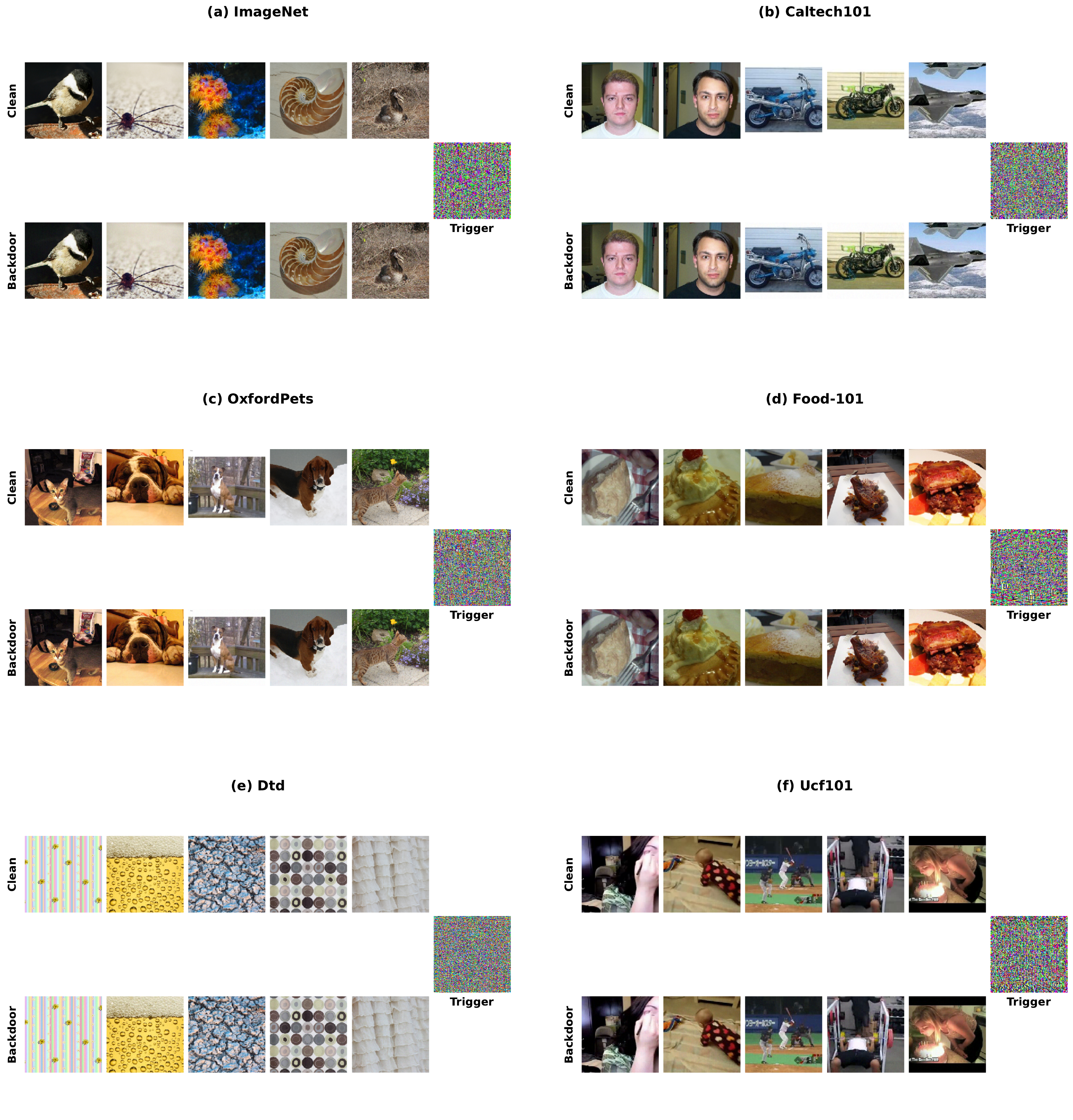}
    \caption{Visualization of clean images (top row of each dataset) and their corresponding backdoor images (bottom row). The trigger (right) is applied to create the backdoor images, but it is visually imperceptible, highlighting the stealthiness of our attack.}
    \label{fig:vis_summary}
\end{figure*}

\section{Theoretical Proof of Functional Fusion under Joint Optimization}

\subsection{Preliminaries and Problem Setup}

Let $f_\theta$ denote the frozen ViT backbone and $g_\phi$ the dynamic Visual Prompt Generator (VPG) with parameters $\phi \in \mathbb{R}^p$.  
For an input $x$, the VPG produces a conditional prompt $\Delta x = g_\phi(h(x))$ based on an intermediate feature $h(x)$.

We consider two tasks optimized jointly:
\begin{align}
    L_{\mathrm{clean}}(\phi) &= \mathbb{E}_{(x,y)\sim\mathcal{D}_{\mathrm{clean}}}
    \, \ell_{\mathrm{CE}}(M_{\theta,\phi}(x), y), \\
    L_{\mathrm{attack}}(\phi,\delta) &= \mathbb{E}_{(x,y)\sim\mathcal{D}_{\mathrm{clean}}}
    \, \ell_{\mathrm{CE}}(M_{\theta,\phi}(T(x,\delta)), y_t),
\end{align}
where $T(x,\delta)$ injects a trigger $\delta$ into $x$ and $y_t$ is the target label.  
The total loss is
\begin{equation}
    L_{\mathrm{total}}(\phi,\delta)
    = L_{\mathrm{clean}}(\phi) + L_{\mathrm{attack}}(\phi,\delta).
    \label{eq:ltotal}
\end{equation}

We assume:
\begin{itemize}
    \item[(A1)] \textbf{Capacity constraint:} VPG has limited parameters (\emph{parameter-efficient fine-tuning}).  
    \item[(A2)] \textbf{Shared representation:} Clean and poisoned inputs share overlapping regions in the backbone feature space $\mathcal{H}$.  
    \item[(A3)] \textbf{Regularization bias:} Optimization implicitly or explicitly favors low-norm (simple) solutions, via weight decay or SGD bias.  
    \item[(A4)] \textbf{Joint objective:} Parameters $\phi$ are optimized by minimizing the combined loss in Eq.~\eqref{eq:ltotal}.
\end{itemize}

Under these conditions, we show that the minimization of $L_{\mathrm{total}}$ inevitably induces a
\emph{Functional Fusion} phenomenon: both benign and malicious functionalities are encoded in the same small subset of parameters, making them inseparable.

\subsection{Linearized Toy Model}

To analyze the phenomenon, we locally linearize $g_\phi$ around its optimum:
\begin{equation}
    g_\phi(h_i) \approx A(h_i) \, \phi,
\end{equation}
where $A(h_i)\in\mathbb{R}^{d_p\times p}$ maps the parameters $\phi$ to a prompt vector.

For $n_c$ clean samples and $n_a$ attack samples, the linearized prediction for each group is:
\begin{align}
    Y_c &\approx B_c + C_c \phi, \\
    Y_a &\approx B_a + C_a \phi,
\end{align}
where $C_c, C_a$ are the respective design matrices, and $B_c, B_a$ denote constant offsets.  
We define the concatenated system:
\begin{equation}
    C = 
    \begin{bmatrix}
        C_c \\ C_a
    \end{bmatrix},
    \qquad
    Y =
    \begin{bmatrix}
        Y_c - B_c \\
        Y_a - B_a
    \end{bmatrix}.
\end{equation}

Using a quadratic (least-squares) approximation to the loss, the total objective with regularization is:
\begin{equation}
    \mathcal{L}(\phi) = \| C\phi - Y \|^2 + \lambda \| \phi \|^2,
    \label{eq:ridge}
\end{equation}
where $\lambda > 0$ captures explicit or implicit regularization.

\subsection{Closed-Form Solution}

\begin{theorem}[Closed-form Solution of Joint Optimization]
The minimizer of Eq.~\eqref{eq:ridge} is
\begin{equation}
    \phi^* = (C^\top C + \lambda I_p)^{-1} C^\top Y.
    \label{eq:ridge-sol}
\end{equation}
\end{theorem}

\begin{proof}
Setting the gradient to zero:
\[
\nabla_\phi \mathcal{L} = 2C^\top(C\phi - Y) + 2\lambda\phi = 0,
\]
which yields $(C^\top C + \lambda I)\phi = C^\top Y$, leading to Eq.~\eqref{eq:ridge-sol}.
\end{proof}

Let $C = U \Sigma V^\top$ be the singular value decomposition (SVD), with
$\Sigma = \mathrm{diag}(\sigma_1,\dots,\sigma_r)$, $r = \mathrm{rank}(C)$.  
Substituting into Eq.~\eqref{eq:ridge-sol} gives
\begin{equation}
    \phi^* = V (\Sigma^2 + \lambda I_r)^{-1} \Sigma U^\top Y.
\end{equation}

The coefficient on the $j$-th singular direction $v_j$ is:
\begin{equation}
    \alpha_j = v_j^\top \phi^* = \frac{\sigma_j}{\sigma_j^2 + \lambda} (u_j^\top Y).
    \label{eq:alphaj}
\end{equation}

Thus, the squared norm decomposes as
\begin{equation}
    \|\phi^*\|^2 = \sum_{j=1}^{r} 
    \left(\frac{\sigma_j}{\sigma_j^2+\lambda}\right)^2 (u_j^\top Y)^2.
    \label{eq:energy}
\end{equation}

\subsection{Energy Concentration on Shared Directions}

Consider the partition $C = [C_c; C_a]$.
The Gram matrix of the system satisfies:
\begin{equation}
    C^\top C = C_c^\top C_c + C_a^\top C_a.
    \label{eq:gram}
\end{equation}

For any unit vector $v \in \mathbb{R}^p$, we have
\begin{equation}
    v^\top (C^\top C) v = \|C v\|^2 = \|C_c v\|^2 + \|C_a v\|^2.
    \label{eq:rayleigh}
\end{equation}
If there exists a direction $v$ such that both $\|C_c v\|$ and $\|C_a v\|$ are large,
then this direction contributes to a high Rayleigh quotient,
therefore becomes one of the dominant singular directions of $C$.
According to Eq.~\eqref{eq:alphaj}–\eqref{eq:energy}, such directions receive large coefficients $\alpha_j$,
causing $\phi^*$ to concentrate its energy on them.

\begin{lemma}[Energy Concentration]
If $C_c$ and $C_a$ have overlapping column spaces (i.e., share non-trivial common directions in parameter space),
then the joint matrix $C$ exhibits enlarged singular values along those directions, 
and the optimal $\phi^*$ allocates most of its $\ell_2$ energy on them.
\end{lemma}

\begin{proof}[Proof Sketch]
Let $\mathcal{S}_c = \mathrm{span}(C_c)$ and $\mathcal{S}_a = \mathrm{span}(C_a)$.
If $\mathcal{S}_c \cap \mathcal{S}_a \neq \emptyset$, then
for any $v \in \mathcal{S}_c \cap \mathcal{S}_a$, both terms in Eq.~\eqref{eq:rayleigh} are positive,
yielding larger eigenvalues of $C^\top C$ in those directions.
From Eq.~\eqref{eq:alphaj}, larger $\sigma_j$ implies stronger coefficients $\alpha_j$.
Hence, $\phi^*$ accumulates energy along shared directions.
\end{proof}

This means that \textbf{the regularized least-squares solution
compresses both tasks' representations into the same low-dimensional subspace}.
This subspace corresponds to a small subset of parameter coordinates with large magnitude—empirically observed as the ``core'' or
``functional fusion'' region.

\subsection{Impact of Pruning and Inseparability}

Partition $\phi = [\phi_S; \phi_{\bar S}]$ and $C = [C_S\,\,C_{\bar S}]$.
If we prune parameters by enforcing $\phi_S = 0$ and re-optimize the remaining $\psi = \phi_{\bar S}$,
the optimal $\psi^*$ satisfies:
\begin{equation}
    (C_{\bar S}^\top C_{\bar S} + \lambda I)\psi^* = C_{\bar S}^\top Y.
\end{equation}

The increase in loss relative to the full optimum $\phi^*$ is:
\begin{equation}
    \Delta\mathcal{L} =
    ([0;\psi^*] - \phi^*)^\top (C^\top C + \lambda I)
    ([0;\psi^*] - \phi^*),
    \label{eq:deltaL}
\end{equation}
which is strictly positive and grows with $\|\phi_S^*\|$.
Thus, removing high-energy coordinates (those forming the ``core'' $S$)
simultaneously harms both $L_{\mathrm{clean}}$ and $L_{\mathrm{attack}}$,
as they are jointly embedded in $C$ and $Y$.

\begin{corollary}[Functional Inseparability]
Let $S$ denote the subset of coordinates where $\phi^*$ has dominant energy.
If $\mathcal{S}_c \cap \mathcal{S}_a \neq \emptyset$,
then pruning or perturbing $S$ increases both $L_{\mathrm{clean}}$ and $L_{\mathrm{attack}}$,
leading to simultaneous degradation of benign accuracy and attack success rate.
\end{corollary}

\subsection{Conclusion (Functional Fusion Theorem)}

\begin{theorem}[Functional Fusion under Joint Optimization]
Under assumptions (A1)–(A4), the joint minimization of
$L_{\mathrm{total}} = L_{\mathrm{clean}} + L_{\mathrm{attack}}$
with regularization inevitably yields a low-norm solution $\phi^*$ 
whose energy is concentrated on the overlapping subspace of
$\mathrm{span}(C_c)$ and $\mathrm{span}(C_a)$.
Consequently, the clean and attack functionalities are encoded in the same small parameter subset,
and any removal of this subset jointly destroys both functionalities.
\end{theorem}

\begin{proof}
Immediate from Eq.~\eqref{eq:rayleigh}–\eqref{eq:energy} and the preceding lemma:
joint optimization amplifies shared directions (via gradient accumulation),
while the regularization $\lambda\|\phi\|^2$ penalizes redundant or disjoint bases.
Hence, the optimizer prefers to reuse shared parameter directions to minimize both losses simultaneously.
This coupling compresses both task mappings into a single low-dimensional core,
yielding the observed functional fusion phenomenon. \qedhere
\end{proof}

\subsection{Remarks}

\begin{itemize}
    \item \textbf{Why dynamic prompts amplify fusion:}
    The conditional mapping $g_\phi(h)$ uses the same parameters $\phi$ to handle diverse inputs.
    Under limited capacity, optimization reuses the same parameter directions
    to realize both benign and malicious behaviors, enhancing fusion.
    \item \textbf{Implications for pruning defenses:}
    Eq.~\eqref{eq:deltaL} shows that removing shared-core parameters increases both losses quadratically,
    explaining the empirical collapse of clean accuracy and ASR simultaneously.
    \item \textbf{Extensions:}
    A full-rank matrix analysis can formally prove that 
    the minimal Frobenius-norm solution prefers shared bases between tasks,
    connecting this phenomenon to classical multi-task learning theory.
\end{itemize}

\subsection{Summary}

The above analysis demonstrates, from an optimization and linear algebra perspective,
that under capacity constraints and regularization,
the joint minimization of $L_{\mathrm{clean}} + L_{\mathrm{attack}}$
inevitably drives the model toward a \emph{functional fusion} regime,
where both tasks share the same compact parameter core.
This explains the empirical inseparability of clean and malicious behaviors
and the failure of pruning-based defenses observed in VIPER.